\documentclass{article}

\usepackage{arxiv}

\usepackage[T1]{fontenc}
\usepackage[utf8]{inputenc}
\usepackage[english]{babel}
\usepackage{graphicx}
\usepackage{amsmath}
\usepackage{amsfonts}
\usepackage{hyperref}
\usepackage{fancyhdr}
\usepackage{lmodern}
\usepackage{cancel}
\usepackage{amssymb}
\usepackage{amsthm}
\usepackage{pgfplots}
\usepackage{accents}
\usepackage{bbm}
\usepackage{bm}
\usepackage{mathrsfs}
\usepackage{booktabs}
\usepackage{url}

\newcommand{\R}{\mathbb{R}}
\newcommand{\C}{\mathbb{C}}
\newcommand{\N}{\mathbb{N}}

\newcommand{\Ra}{\mathrm{Ra}}
\newcommand{\Da}{\mathrm{Da}}
\renewcommand{\Pr}{\mathrm{Pr}}

\renewcommand{\bar}{\overline}
\renewcommand{\rho}{\varrho}

\newcommand{\be}{\begin{equation}}
\newcommand{\ee}{\end{equation}}

\theoremstyle{plain}
\newtheorem{thm}{Theorem}[section]

\newtheorem{lem}{Lemma}[section]

\theoremstyle{definition}

\title{On the competition between rotation and variable viscosity in a Darcy-Brinkman model} 

\author{  
F. Capone\thanks{Corresponding author.} \href{https://orcid.org/0000-0002-0672-999X}{\includegraphics[scale=0.1]{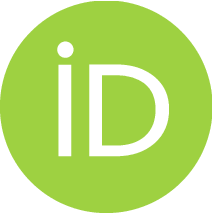}} \\ Dipartimento di Matematica e Applicazioni 'R.Caccioppoli' \\ Università degli Studi di Napoli Federico II \\ Via Cintia, Monte S.Angelo, 80126 Napoli \\ Italy \\ 
\texttt{fcapone@unina.it} \\
\And 
J.A. Gianfrani\href{https://orcid.org/0000-0001-9906-2495}{\includegraphics[scale=0.1]{orcid.eps}} \\ Research Centre for Fluid and Complex Systems \\ Coventry University \\ Priory St, Coventry CV1 5FB \\ United Kingdom \\   
\texttt{ae2688@coventry.ac.uk} }

\renewcommand{\shorttitle}{}

\begin{document}
\maketitle

\begin{abstract}
In the present paper, the onset of thermal convection in a uniformly rotating Darcy-Brinkman porous medium saturated by a variable viscosity fluid is investigated and the competing interplay between rotation and temperature-dependent viscosity is then analysed. 
In literature it has been proved that considering a variable viscosity fluid saturating a porous medium does not produce any additional oscillating motions at the onset of convection. In this direction, similar results have been obtained in \cite{CaponeGentile2000}, where the authors prove the validity of the principle of exchange of stabilities in a rotating porous medium. Here, it is shown numerically that the combined effect of rotation and variable viscosity  may lead to the occurrence of oscillatory convection. This result is remarkable and not known in literature, to the best of our knowledge. The stabilising effect of rotation on the onset of convection is recovered, while it is shown that the effect of variable viscosity on the behaviour of the critical Rayleigh number depends on the choice of the Taylor number. Nonlinear analysis is performed via the energy method and sufficient condition for the stability of the basic solution is determined. Proximity of the results from the linear and nonlinear analyses is obtained thanks to numerical techniques, i.e. Chebyshev-$\tau$ method and golden section method.
\end{abstract}
\keywords{Porous Media \and Variable viscosity \and Oscillatory convection \and Conditional stability \and Rotating flow }

\section{Introduction}
Instability in rotating flows is a wide topic of research that has attracted the attention of many researchers over the years. Many mathematical studies have been performed to understand the physical mechanism that leads to the occurrence of certain flow structures. To mention a few examples, instability in rotating flows has been detected in cylindrical enclosures leading to the so-called centrifugal instability \cite{chandra, drazin, taylor1923, park2013stably, park2017instabilities},  but also in horizontal layers where a uniform temperature gradient is present across the layer \cite{hughes, vadasz2021centrifugal,straughan2023rotating, Gianfrani2023}. In this context, much of attention has been devoted to the occurrence of thermal instability in porous media and in particular to the mathematical consequences of the presence of the skew-symmetric term of Coriolis acceleration. When performing a nonlinear analysis via the energy theory, one has to pay attention to the choice of the Lyapunov functional given that a standard choice would lead to oversimplification of the problem, implying a loss of information from the skew-symmetric term. Therefore, alternative procedures have been developed during the years, to take into account the presence of the Coriolis term within the analysis \cite{CaponeGentile2000, Capone2020, vadasz1998, vadasz2019instability, gianfrani2020, capone2022thermal}.
In \cite{CaponeGentile2000}, the authors modelled a physical problem in a uniformly rotating porous medium, where the onset of oscillating motions are ruled out in the diffusive regime for Prandtl number greater than 1 ($\Pr\geq 1$). Moreover, the authors studied the nonlinear stability and managed to find an appropriate Lyapunov functional that could lead to the coincidence between the results coming from the linear and nonlinear stability analyses. These results are quite remarkable in the context of rotating flows in porous media given that the paper from Vadasz in 1998 \cite{vadasz1998} seemed to rule out the possibility of coincidence of results between the two analyses.
Moreover, from a physical point of view, the authors find confirmation of the delaying effect of axial rotation on the onset of thermal instability. This is a well-known mechanism for which fluid particles are pushed outward by centrifugal forces that oppose the buoyancy force, resulting in a delaying effect of rotation on the onset of convection. 

The opposite effect is witnessed by dynamic viscosity as a function of temperature such that viscosity is greater at lower temperatures and smaller at higher ones  \cite{torrance, straughan1986, kassoy1975, acta2022}. Viscosity is a measure of the ease of a fluid to move, therefore, once a downward gradient of temperature is imposed across a porous layer, fluid viscosity will be lower at the bottom, facilitating the action of the buoyancy force and the onset of thermal instability. The decreasing trend of viscosity with temperature is typical of liquids \cite[p. 253]{vafai, NB}, for which normally temperature dependence of viscosity is modelled with Arrhenius law. For many other liquids, experimental data are well-approximated by an exponential law of this type:
\begin{equation}\label{exp}
    \mu(T)= \mu_0 F(T) \quad \text{where } F(T)=\text{exp}(-\gamma T)
\end{equation}
$\mu, T$ being dynamic viscosity and temperature, respectively.
It is the case of engine oils, for which an exponential decrease of viscosity with rising temperatures occurs. In order to fit experimental data, some values of coefficients in Eq. \eqref{exp}  are shown in Table \ref{tab1}:
\begin{table}[htbp]
  \centering
  \begin{tabular}{|c|c|c|}
    \hline
    \text{Oil name} & $\mu_0$ [mPI] & $\gamma$ [$^\circ C^{-1}$] \\
    \hline
    SAE15W-40 & 1107.65 & 0.065  \\
    SAE5W-40 & 649.64 & 0.055 \\
    SAE10W-60 & 1250.60 & 0.057 \\
    \hline
  \end{tabular}
  \caption{Coefficients obtained from  experimental data available in \cite{website2}, regarding different engine oils.}
  \label{tab1}
\end{table}

Cooling liquids, such as oils or a solution of water and glycol are very often used within engines. The choice of glycol is based on the fact that this liquid lowers the freezing point of water and increases the boiling point, so that the fluid can work in a wider interval of temperatures (see data available in \cite{website1}). On the other hand, engine oils are well-suitable to get around high pressure issues in high temperature contexts and keep parts of the engine lubricated. Moreover, engine oils are responsible for a large percentage of the cooling that takes place within the engine components, where most of the heat is produced because of combustion and friction.
As oil passes through the engine, it is directed onto these hot components in order to carry the heat away to the oil sump, where, from here the heat is dissipated to the air surrounding the sump.
Therefore, within an engine, it is crucial that the cooling fluid is kept at the right viscosity so that it can circulate easily and cool down each component of the engine. For this reason, given that viscosity is affected by temperature, it is important to maintain the liquid at the optimal temperature for an efficient cooling process. At high temperatures, fluid viscosity is lower and the efficiency of the cooling system is higher.

In literature, thermal convection in variable viscosity fluids has been extensively studied since the works from \cite{torrance, joseph}. In \cite{rajagopal2009stability} the authors considered an exponential dependence of viscosity with respect to temperature and pressure. They studied the effect of exponential variations of viscosity on the onset of instability, proving the coincidence between linear and nonlinear analysis results. Whereas, in \cite{straughan2023nonlinear}, the author studied the case of a non-Newtonian fluid whose viscosity is a quadratic function of temperature. The interesting result lies in the fact that the nonlinear stability thresholds depend on the size of the initial
temperature perturbation.

The present paper is devoted to investigate the interplay between two competing effects: uniform axial rotation and  variable viscosity. The remarkable result coming from the linear analysis involves the lack of validity of the principle of exchange of stabilities as a consequence of the interaction between the two effects. There are numerous studies of thermal instability in porous media saturated by a variable viscosity fluid where the hypothesis of variable viscosity does not affect the type of motion originating at the onset. The reader can find some examples in \cite{straughan1986, acta2022, rajagopal2009stability}.
In the present paper, it has been proved numerically that the leading eigenvalues at the critical point may have nonzero imaginary part, leading to the occurrence of oscillatory convection and wave-like motion within the fluid, and this is due to the presence of nonzero viscosity gradient. To the best of our knowledge this result is not known in literature and may open new perspective regarding the interplay between competing hypotheses within the physical setup. Moreover, it has been possible to determine a non-trivial Lyapunov functional that leads to a stability theorem and stability thresholds that are very close to the linear ones, which represents the best result one could expect from a nonlinear energy analysis.

The plan of the paper is the following. Section \ref{sec1} is devoted to the explanation of the physical setup and determination of the mathematical model in its dimensional and dimensionless form. Modelling approximations are reported and fundamental dimensionless numbers are defined. In section \ref{sec2}, the linear stability analysis is performed and the numerical technique employed to solve the differential eigenvalue problem is introduced. Whereas, in section \ref{sec3}, the nonlinear stability analysis is reported. The main steps of the analysis are the determination of a suitable energy functional, the control of production and nonlinear term in order to recover the exponential decay in time of the total energy and finally the resolving the Euler-Lagrange equations to determine the critical Rayleigh number for nonlinear stability. Section \ref{sec4} is devoted to present numerical results from linear and nonlinear analyses, obtained by means of numerical techniques, i.e. Chebyshev-$\tau$ method and golden section method.

\section{Mathematical Model}\label{sec1}

Let $Oxyz$ be a Cartesian frame of reference where the $z$-axis is vertically upward. Let us consider a horizontal isotropic porous layer delimited by two impervious planes ($z=0$ and $z=d$) and saturated by a Newtonian fluid at rest  (see Fig. \ref{rotating_frame}). Planes confining the layer are kept at a constant temperature at any time in such a way that the fluid-saturated porous medium is heated from above. As a consequence, a uniform gradient of temperature is imposed and maintained constant across the medium. Let $T_L$ be the temperature on the
lower plane $z = 0$ delimiting the layer, and let $T_U$ be the temperature on the upper plane $z = d$. We assume the
layer to be heated from below, i.e. $T_L > T_U$.

Moreover, the layer is uniformly rotating about the vertical axis $z$ (see Fig. \ref{rotating_frame}). Therefore, we assume the reference frame to be the rotating frame in such a way that the modification of the Darcy's law includes a term due to Coriolis acceleration and a new pressure, called reduced pressure, that includes the centrifugal acceleration. As pointed out in \cite{vadasz1998}, not far from the axis of rotation, the centrifugal buoyancy can be neglected compared to the gravity buoyancy, limiting the effect of rotation to the Coriolis term.

In the present paper, based on previous considerations, we relax one point of the Oberbeck-Boussinesq approximation and consider a temperature-dependent fluid viscosity in the Darcy's law. 
The choice of which kind of functional dependence is more appropriate to consider is based on experimental evidence. For the problem at stake, Eq. \eqref{exp} is employed to model viscosity variations.

\begin{figure}
    \centering
    \includegraphics[scale=0.5]{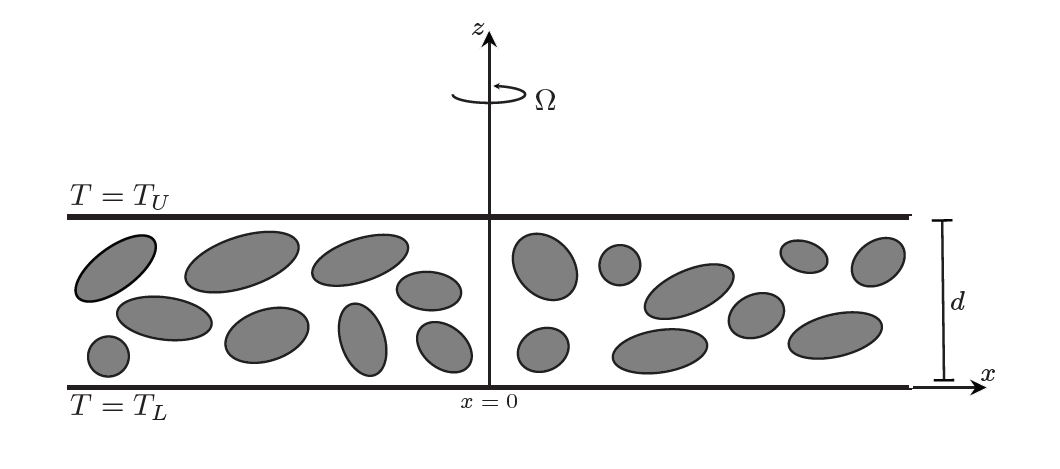}
    \caption{Rotating porous layer}
    \label{rotating_frame}
\end{figure}

Therefore, the mathematical model governing the evolution of velocity, temperature and pressure fields $(\textbf{v}, T, p)$ is, \cite{CaponeGentile2000, NB}:
\begin{equation}\label{mod}
		\begin{cases}
			\frac{\rho_f}{\varepsilon}  \textbf{v}_{,t} = -\nabla p^\prime + \rho_f \alpha g T \textbf{k} - \frac{\mu(T)}{K}\textbf{v}-\frac{2\rho_f \Omega}{\varepsilon}\textbf{k}\times\textbf{v}+\mu_0\Delta \textbf{v}\\
			\nabla \cdot \textbf{v}=0\\
			T_{,t} + \textbf{v} \cdot \nabla T =  \kappa \Delta T 
		\end{cases}
	\end{equation}
where $p^\prime = p-\frac{\rho_f}{2}| \Omega\textbf{k}-\textbf{x}|^2$ is the reduced pressure and $\varepsilon$, $\rho_f$, $c$ and $\kappa$ are medium porosity, density, specific heat and thermal diffusivity.

Boundary conditions attached to model \eqref{mod} are
\begin{equation}\label{bc1}
\begin{split}
    &T=T_L\quad \text{on } z=0   \\
    &T=T_U \quad \text{on } z=d\\ 
    &\textbf{v}\cdot \textbf{n}=0 \quad \text{on } z=0,d
\end{split}
\end{equation}
where Eq. \eqref{bc1}$_3$ models the impervious planes, being $\textbf{n}$ the outward unit normal to planes $z=0, d$. 

The previous differential problem admits steady conductive solution
\begin{equation}
    m_b = \{ \textbf{v}_b, T_b(z), p_b(z)\}
\end{equation}
where $\textbf{v}_b=\textbf{0}$, $T_b(z) = -\beta z +T_L$
and $p_b(z)=\rho_f\alpha g (-\beta \frac{z^2}{2}+T_L z)+cost$.

 In order to study the instability of this solution, we introduce a perturbation $\pi, \theta, \textbf{u}$ on pressure, temperature and velocity, respectively.
 The originating solution will be
 \begin{equation}
    \textbf{v} = \textbf{u} + \textbf{v}_b \qquad  T = \theta + T_b \qquad p^\prime = \pi + p_b
\end{equation}
   Upon plugging this solution into system \eqref{mod}, the following differential problem that governs the evolution of perturbation fields is obtained
\begin{equation}\label{syst}
\begin{cases}
			\frac{\rho_f}{\varepsilon}  \textbf{u}_{,t} = -\nabla \pi + \rho_f \alpha g \theta \textbf{k} - \frac{\mu(T_b+\theta)}{K}\textbf{u}-\frac{2\rho_f \Omega}{\varepsilon}\textbf{k}\times\textbf{u}+\mu_0\Delta \textbf{u}\\
			\nabla \cdot \textbf{u}=0\\
			\theta_{,t} + \textbf{u} \cdot \nabla \theta = \beta w+ \kappa \Delta \theta 
		\end{cases}
	\end{equation}
with boundary conditions:
\begin{equation}
    w=\theta=0 \quad \text{on } z=0,d
\end{equation}

A good compromise between mathematical tractability and still good approximation of the real phenomena involves  approximating the exponential function $F(T)$ with a Taylor series arrested at first order \cite{rajagopal2009stability}. In such a way, we get
 \begin{equation}
     F(T_b+\theta)= F(T_b)+\frac{dF}{d\xi}\biggr|_{\xi=T_b}\theta 
 \end{equation}
 By definition of the basic state 
 \begin{equation}
     F(T_b+\theta)= \tilde{f}(z)-\gamma \tilde{f}(z)\theta \quad \text{where } \tilde{f}(z)=\text{exp}(-\gamma (-\beta z+T_L)) 
 \end{equation}
Hence, Eq. \eqref{syst} becomes
\begin{equation}\label{syst2}
\begin{cases}
			\frac{\rho_f}{\varepsilon}  \textbf{u}_{,t} = -\nabla \pi + \rho_f \alpha g \theta \textbf{k} - \frac{\tilde\mu_0 f(z)}{K}\textbf{u} +\frac{\tilde\mu_0 \gamma f(z)}{K}\textbf{u}\theta-\frac{2\rho_f \Omega}{\varepsilon}\textbf{k}\times\textbf{u}+\mu_0\Delta \textbf{u}\\
			\nabla \cdot \textbf{u}=0\\
			\theta_{,t} + \textbf{u} \cdot \nabla \theta = \beta w+ \kappa \Delta \theta 
		\end{cases}
	\end{equation}
where, for the sake of simplicity, $f(z)=\text{exp}(\gamma\beta(z-d/2))$, $\tilde\mu_0 = \mu_0 \text{exp}(-\gamma(T_L+T_U)/2)$.

In Eq. \eqref{syst2}, the layer depth $d$ together with the following parameters
\begin{equation}\label{nondim2}
\begin{split}
  \tau=\dfrac{\rho_f d^2}{\tilde\mu_0 \varepsilon}, \quad U=\dfrac{\varepsilon \tilde\mu_0}{\rho_f d},\quad  P=\dfrac{\tilde\mu_0 U}{d} ,\quad T^{\#}=\dfrac{\beta d^2 U}{\kappa \Ra} .
\end{split}
\end{equation}
is used to scale the dimensional unknown fields.
Hence, the following set of variables in introduced, where the asterisks denote nondimensional fields
\begin{equation}\label{nondim1}
    \begin{split}
 &\textbf{x}=d\textbf{x}^*,\quad t=\tau t^*\\
 &\textbf{v}=U\textbf{v}^*,\quad \pi=P \pi^*,\quad T=T^{\#}T^*,
    \end{split}
\end{equation}
The resulting nondimensional version of Eq. \eqref{syst2} is
 \begin{equation}\label{modad}
\begin{cases}
			 \textbf{u}_{,t} = -\nabla \pi + \Ra \theta \textbf{k} - \frac{f(z)}{\Da}\textbf{u} + \frac{\hat\gamma f(z)}{ \Da} \theta \textbf{u} -\mathcal{T}\textbf{k}\times\textbf{u}+\Delta \textbf{u}\\
			\nabla \cdot \textbf{u}=0\\
			\Pr \theta_{,t} + \Pr\textbf{u} \cdot \nabla \theta = \Ra w+  \Delta \theta 
		\end{cases}
	\end{equation}
with $f(z)=\text{exp}(\Gamma(z-1/2))$ being $\Gamma=\gamma\beta d$ and
where $\hat\gamma=T^{\#} \gamma$ and
 \begin{equation}\label{def_Ra}
 \Ra=\sqrt{\dfrac{g\rho_f\alpha \beta d^4}{\tilde\mu_0\kappa}},\quad \Da=\dfrac{k}{d^2},\quad \mathcal{T}=\dfrac{2\Omega \rho_f d^2}{\varepsilon\tilde\mu_0}, \quad \Pr=\frac{\varepsilon\tilde\mu_0}{\rho_f\kappa}.
\end{equation}
 are the Rayleigh number, the Darcy number, the Taylor number and the Prandtl number, respectively.



By retaining the third component of the curl and the double curl of Eq. \eqref{modad}$_1$, it turns out that
 \begin{equation}\label{doublecurl}
\begin{cases}
			 -\Delta w_{,t} = -\Delta\Delta w - \Ra \Delta_1\theta + \frac{f(z)}{\Da}\Delta w +\frac{f'(z)}{\Da} w_{,z}   + \nabla\times\nabla\times(\frac{\hat\gamma f(z)}{ \Da} \theta \textbf{u})\cdot\textbf{k} +\mathcal{T}\omega_{,z}\\
			\omega_{,t} =\Delta \omega -\frac{f(z)}{\Da}\omega+\mathcal{T}w_{,z} + \nabla\times(\frac{\hat\gamma f(z)}{ \Da} \theta \textbf{u})\cdot\textbf{k} \\
			\Pr \theta_{,t} + \Pr\textbf{u} \cdot \nabla \theta = \Ra w+  \Delta \theta 
		\end{cases}
	\end{equation}
 with boundary conditions
 \begin{equation}
    w=\theta=\omega_{,z}=0 \quad \text{on } z=0,1
\end{equation}

We shall now assume that perturbations are periodic in the $x$ and $y$ direction with periods $\frac{2\pi}{k_x}$ and $\frac{2\pi}{k_y}$, respectively. Moreover, $\forall g\in\{\textbf{u}, \omega, \theta\}$
\begin{equation}
    g \ : (\textbf{x},t) \in \Omega\times \R^+ \rightarrow g(\textbf{x},t)\in\R \ \text{and } g\in W^{3,2}(\Omega) \ \forall t\in\R^+,
\end{equation}
where $\Omega$ is the periodicity cell
\begin{equation}
    \Omega = \left[0,\frac{2\pi}{k_x}\right]\times\left[0,\frac{2\pi}{k_y}\right]\times[0,1]
\end{equation}
and $g$ can be expanded in a Fourier series uniformly convergent in $\Omega$.

\section{Linear analysis}\label{sec2}
The present section addresses the onset of stationary and oscillatory convection via a linear stability analysis of the conduction solution. 
The linear version of \eqref{doublecurl} is the following
 \begin{equation}\label{linsys}
\begin{cases}
			 -\Delta w_{,t} = -\Delta\Delta w - \Ra \Delta_1\theta + \frac{f(z)}{\Da}\Delta w +\frac{f'(z)}{\Da} w_{,z}  +\mathcal{T}\omega_{,z}\\
			\omega_{,t} = \Delta \omega -\frac{f(z)}{\Da}\omega+\mathcal{T}w_{,z}  \\
			\Pr \theta_{,t} = \Ra w+  \Delta \theta 
		\end{cases}
\end{equation}
with boundary conditions
 \begin{equation}
    w=\theta=\omega_{,z}=0 \quad \text{on } z=0,1
\end{equation}

Since system \eqref{linsys} is autonomous and perturbations are assumed periodic in $x$ and $y$ directions of periods $\tfrac{2\pi}{k_x}$ and $\tfrac{2\pi}{k_y}$ respectively, solution of Eq. \eqref{linsys} can be written as superposition of modes 
\begin{equation}\label{visc_expansion}
	\varphi(t,x,y,z) = \sum_{n=1}^{+\infty} \bar{\varphi}_n(x,y,z) e^{\sigma t} \quad (\sigma\in\C)  \quad \forall \varphi \in \{w, \omega, \theta \}
\end{equation}
where 
\begin{equation}
	\Delta_1 \bar{\varphi}_n(x,y,z) = - k^2 \bar{\varphi}_n(z)  \quad\qquad k^2=k^2_x + k^2_y
\end{equation}
and $\sigma$ is the temporal growth rate and $k^2$ is the wavenumber of the perturbation.

According to the linear theory, the basic conduction solution is unstable when the real part of $\sigma$ is positive ($Re(\sigma)>0$). A primary instability occurs leading to the origin of a new steady state if, at the onset, $Im(\sigma)=0$. Otherwise, a wave-like motion emerges within the fluid and the so-called oscillatory convection occurs.

Upon substituting Eq. \eqref{visc_expansion} into Eq. \eqref{linsys} and retaining only the $n$-th by virtue of linearity of the system, the following differential eigenvalue problem is obtained:
\begin{equation}\label{modad2}
\begin{cases}
    0 =(D^2-k^2) w -\Phi\\
    \sigma \Phi = (D^2-k^2)\Phi -P_0f(z)\Phi -P_0 f(z) \Gamma Dw-\Ra  k^2 \theta - \mathcal{T}  D\omega\\
    \sigma  \omega = (D^2-k^2) \omega -P_0 f(z) \omega +\mathcal{T}  D w\\
    \Pr \sigma \theta =\Ra w + (D^2-k^2)\theta
\end{cases} 
\end{equation}
where  $D:=\frac{d}{dz}$ and $P_0=\Da^{-1}$, with boundary conditions
 \begin{equation}\label{bc5}
    w=\Phi=\theta=D\omega=0 \quad \text{on } z=0,1
\end{equation}
Hence, Eqs. \eqref{modad2}-\eqref{bc5} represent a differential eigenvalue problem of this kind
\begin{equation}
    \mathcal{A}\textbf{X}=\sigma\mathcal{B}\textbf{X} \qquad \textbf{X}=(w,\Phi,\omega,\theta).
\end{equation}

The presence of non-constant coefficients within the model makes a pure analytical procedure very complex. Therefore, in order to determine the eigenvalue $\sigma$ and consequently the locus of couples $(k,\Ra^2)$ where $Re(\sigma)=0$,  a Chebyshev-$\tau$ method has been implemented. 
Details on the method are provided in Appendix \ref{appendix1}. For a further discussion the reader is referred to \cite{dongarra, arnone2023chebyshev, bourne2003,arnonegianfrani}.

 \section{Nonlinear analysis}\label{sec3}
In the present section, using the well-established energy method \cite{straughan2013, Galdi1985, galdi1985exchange, straughan2006global, barletta, gianfrani2022}, and following guidelines provided in \cite{CaponeGentile2000, RioneroMulone1987}, it is reported a nonlinear stability analysis upon restriction on initial data of the total energy associated to the nonlinear system under the spotlights. This analysis leads to a fundamental stability theorem by which it is possible to claim that perturbations on the conduction solution do not trigger convection for low values of the Rayleigh number, provided that a constraint on the initial datum of the energy is verified. Moreover, the critical Rayleigh number for nonlinear stability is determined, showing good agreement with the results coming from the linear analysis. The proximity of the results from the two analyses confirms the validity of the Lyapunov functional chosen to undertake the nonlinear analysis.  
 
\subsection{Essential variables and Lyapunov functional}
In the linear stability analysis for the conductive solution, the variables employed are $w,\theta,\omega$, perturbation fields on the vertical component of the seepage velocity, temperature and third component of flow vorticity, respectively. Following the nomenclature provided in \cite{RioneroMulone1987}, these variables are called essential variables.

 In order to study the nonlinear stability of the same solution, we need to introduce a new essential variable. This new variable is a combination of vorticity and temperature, modelling a balance between those variables who promote or inhibit instability. Following the choice done in \cite{CaponeGentile2000}, we define
 \begin{equation}
     \varphi = \Ra\Delta_1\theta -\mathcal{T} \omega_{,z}
 \end{equation}
 It turns out that the nonlinear system governing the evolution of perturbations is 
 \begin{equation}\label{nonlinsystem}
\begin{cases}
			 -\Delta w_{,t} = -\Delta\Delta w + \frac{f(z)}{\Da}\Delta w +\frac{f'(z)}{\Da} w_{,z}   + \nabla\times\nabla\times(\frac{\hat\gamma f(z)}{ \Da} \theta \textbf{u})\cdot\textbf{k} -\varphi\\
    \varphi_{,t} = q\Ra^2 \Delta_1 w + q\Ra \Delta\Delta_1\theta -\mathcal{T}\Delta \omega_{,z}-\mathcal{T}^2w_{,zz}+ \frac{\mathcal{T}}{\Da}( f \omega_{,z} +f^\prime \omega) \pm \Ra\Delta\Delta_1\theta\\
    \qquad\qquad\qquad\qquad\qquad\qquad\qquad\qquad-\Ra\Delta_1(\textbf{u}\cdot\nabla\theta)-\mathcal{T} \partial_{,z}\nabla\times(\frac{\hat\gamma f(z)}{ \Da} \theta \textbf{u})\cdot\textbf{k}\\
			\omega_{,t} =\Delta \omega -\frac{f(z)}{\Da}\omega+\mathcal{T}w_{,z} + \nabla\times(\frac{\hat\gamma f(z)}{ \Da} \theta \textbf{u})\cdot\textbf{k} \\
			\Pr \Delta_1\theta_{,t} + \Pr\Delta_1(\textbf{u} \cdot \nabla \theta) = \Ra \Delta_1w+  \Delta\Delta_1 \theta 
		\end{cases}
	\end{equation}
 where $q=\Pr^{-1}$, with boundary conditions:
 \begin{equation}\label{bcNL}
     w=\omega_{,z}=\theta=\Delta_1\theta=\theta_{,zz}=\varphi=0 \ \text{on } z=0,1
 \end{equation}

By denoting with $(\cdot,\cdot)$ and $\|\cdot\|$, respectively, the scalar product and the associated norm in $L^2(V)$, one can obtain the evolution equation of the total energy associated to system  \eqref{nonlinsystem}-\eqref{bcNL}. With this aim, let us introduce the following Lyapunov functional 
 \begin{equation}
    E=E_1+bE_2
\end{equation}
where 
 \begin{equation}
     E_1 = \frac{1}{2} \{\lambda_1 \|\nabla w\|^2 + \|\varphi\|^2 +\mu\Pr\|\Delta_1\theta\|^2+\lambda_2\|\omega\|^2\}
 \end{equation}
 and \begin{equation}
    E_2= \frac{1}{2} \{\chi\|\nabla\textbf{u}\|^2 + \Pr \|\Delta\theta\|^2 +\|\nabla\omega\|^2+\|\textbf{u}\|^2+\|\omega\|^2+ \chi \|\Delta \textbf{u}\|^2\}
\end{equation}
where $\lambda_1, \lambda_2, \mu$ are three coupling parameters to be optimally chosen later, while $\chi$ is a constant that will be conveniently determined in due course.

\subsection{Control of production and nonlinear terms}
The evolution of the total energy $E$ is governed by the following equation
\begin{equation}\label{ineq1}
    \frac{dE}{dt}= I_1-D_1+N_1+bI_2 -bD_2 +bN_2\leq -D_1(1-m) +bI_2-bD_2+N_1+bN_2
\end{equation}
where by definition
\begin{equation}\label{varprob}
     m=\max_{\mathcal{H}} \mathcal{F}, \qquad \text{being } \mathcal{F}= \frac{I_1}{D_1}
 \end{equation} 
 where the space of kinematically admissible perturbations $\mathcal{H}$ is defined as:
\begin{equation}
\begin{split}
  \mathcal{H}=\Bigl\{(\textbf{u},\theta,\omega,\varphi)\in W^{3,2}(V)\;&|\;x,y\;\text{periodic} \ \text{with period }2\pi/k_x,2\pi/k_y,\;\text{verifying (\ref{bcNL})}\Bigr\} ,  
\end{split}
\end{equation}
and where the following production, dissipative and nonlinear terms are determined from the nonlinear system \eqref{nonlinsystem}, upon suitable scalar multiplication in $L^2(V)$, 
 \begin{equation}
\begin{aligned}
     &
     \begin{split}
         I_1 & = -\lambda_1 (\varphi,w) + q\Ra^2 (\Delta_1 w,\varphi) -(\Ra(1-q) \Delta\Delta_1 \theta, \varphi) -\mathcal{T}^2(w_{,zz},\varphi) +\frac{\mathcal{T}}{\Da}(\varphi,f\omega_{,z}+f^\prime\omega)\\ &\quad +\mu\Ra(\Delta_1w, \Delta_1\theta)+\lambda_2\mathcal{T}(w_{,z},\omega)
     \end{split}\\
     & D_1 = \lambda_1 \|\Delta w\|^2 + \frac{\lambda_1}{\Da}\int_V f |\nabla w|^2 + \|\nabla \varphi\|^2 + \mu \|\nabla \Delta_1 \theta\|^2 + \lambda_2 \|\nabla\omega\|^2 + \frac{\lambda_2}{\Da} \int_V f|\omega|^2\\
     & 
     \begin{split}
         N_1 & = \lambda_1 (\nabla\times\nabla\times(\frac{\hat\gamma f(z)}{ \Da} \theta \textbf{u})\cdot\textbf{k}, w) - \mu \Pr (\Delta_1(\textbf{u}\cdot \nabla\theta), \Delta_1\theta) -(\Ra\Delta_1(\textbf{u}\cdot\nabla\theta),\varphi)\\
     &\quad -(\mathcal{T} \partial_{,z}\nabla\times(\frac{\hat\gamma f(z)}{ \Da} \theta \textbf{u})\cdot\textbf{k},\varphi)+\lambda_2 (\nabla\times(\frac{\hat\gamma f(z)}{ \Da} \theta \textbf{u})\cdot\textbf{k},\omega)
     \end{split}
 \end{aligned}
 \end{equation}
 and
\begin{equation}
    \begin{aligned}
        & 
        \begin{split}
            I_2 & =-\chi P_0\int_V f^\prime \textbf{u}\cdot \textbf{u}_{,z} -\chi(\Ra\theta, \Delta w) + 2(\Ra\Delta w, \Delta\theta) \\ & \quad -P_0\int_V f^\prime \omega\omega_{,z} + \Ra(\theta,w)+ {\color{black}\chi P_0 \int_V f^\prime \textbf{u}\nabla\Delta\textbf{u}+\chi  P_0 \int_V f\nabla\textbf{u}\nabla \Delta\textbf{u}}
        \end{split}\\
        & 
        \begin{split}
            D_2 & = \chi\|\Delta \textbf{u}\|^2+\int_V \chi P_0 f |\nabla \textbf{u}|^2 + \|\nabla\Delta\theta\|^2+ \|\Delta\omega\|^2+\int_V P_0 f |\nabla \omega|^2+\|\nabla \textbf{u}\|^2 \\ & \quad +\int_V P_0 f | \textbf{u}|^2 + \|\nabla\omega\|^2+\int_V P_0 f |\omega|^2 + \chi \|\nabla\Delta\textbf{u}\|^2
        \end{split}\\
        & 
        \begin{split}
            N_2 & =  \chi P_0\hat\gamma (f\theta \textbf{u},\Delta \textbf{u})-\Pr(\Delta(\textbf{u}\cdot\nabla\theta),\Delta\theta) + P_0(\nabla\times(\hat\gamma f(z) \theta \textbf{u})\cdot\textbf{k},\Delta\omega)   \\
            & \quad +P_0\hat\gamma(f\theta\textbf{u},\textbf{u})+P_0(\nabla\times(\hat\gamma f(z) \theta \textbf{u})\cdot\textbf{k},\omega) {\color{black}- \chi (f^\prime \theta \textbf{u}, \nabla\Delta\textbf{u})-\chi (f\nabla\theta\textbf{u},\nabla\Delta\textbf{u})-\chi (f\theta \nabla\textbf{u}, \nabla\Delta\textbf{u})}
        \end{split}
    \end{aligned}
\end{equation}

In order to reshape inequality \eqref{ineq1}, the following Lemma has to be proved: 

\begin{lem}\label{lemma1}
    For any perturbation $\theta$ in the space of kinematically admissible functions $\mathcal{H}$, the following inequality holds:
    \begin{equation}\label{ineq}
        \|\nabla\theta\|^2 \leq c_P^2\|\nabla\Delta\theta\|^2
    \end{equation}
\end{lem}
\begin{proof}
    By virtue of periodicity and boundary conditions \eqref{bcNL}, the generalised Young and Poincar\'e inequalities, it turns out that:
    \begin{equation}
        \|\nabla\theta\|^2 \leq |(\theta,\Delta\theta)|\leq \epsilon \frac{\|\theta\|^2}{2} + \frac{\|\Delta\theta\|^2}{2\epsilon} \leq \epsilon c_P \frac{\|\nabla\theta\|^2}{2} + c_P\frac{\|\nabla\Delta\theta\|^2}{2\epsilon}
    \end{equation}
    Therefore, Eq.\eqref{ineq} follows upon choosing $\epsilon=\frac{1}{c_P}$.
\end{proof}

Consequently, the following Lemma holds:
\begin{lem}
    Under boundary conditions Eq. \eqref{bcNL}, for any choice of kinematically admissible perturbations in $\mathcal{H}$, the production term $I_2$ is dominated by dissipation terms.\\
    In particular:
\begin{equation}
    bI_2\leq  D_3
\end{equation}
where \begin{equation}\label{d3}
    D_3= \frac{1-m}{2}D_1 + \frac{b}{2}D_2
\end{equation}
\end{lem}
\begin{proof}

Given that $f(z)\in C^0([0,1])$, by virtue of Weierstrass theorem, we can define $\widehat{M}=\max_{[0,1]} f(z)$ and $\widehat{m}=\min_{[0,1]} f(z)$. By virtue of the Cauchy-Schwarz inequality, the arithmetic-geometric mean inequality, the generalised Young inequality, the Poincar\'e inequality and Lemma \ref{lemma1} the following chain of inequalities holds, where the absolute value of scalar products has been omitted out of simplification:
\begin{equation}\label{eq33}
\begin{aligned}
    & 
    \begin{split}
        I_2 \leq |I_2| \leq  & \ \chi P_0\int_V f^\prime \textbf{u}\cdot \textbf{u}_{,z} +\chi(\Ra\theta, \Delta w) + (1+\chi)(\Ra\Delta w, \Delta\theta) +P_0\int_V f^\prime \omega\omega_{,z} + \Ra(\theta,w) \\ &  +\chi P_0 \int_V f^\prime \textbf{u}\nabla\Delta\textbf{u}+ \chi P_0 \int_V f\nabla\textbf{u}\nabla \Delta\textbf{u}
    \end{split} \\
    & \begin{split}
        \leq & \ \chi \Gamma \widehat{M} P_0 \int_V \textbf{u}\cdot \textbf{u}_{,z} + \Gamma \widehat{M} P_0 \int_V \omega\omega_{,z} + \chi\Ra ( w, \Delta\theta) + (1+\chi)\Ra (\Delta w, \Delta \theta) + \frac{\|w\|^2}{2\epsilon}+ \Ra^2c_P\epsilon\frac{\|\nabla\theta\|^2}{2} \\& + \chi P_0\Gamma \widehat{M}\int_V\textbf{u}\nabla\Delta\textbf{u}+\chi P_0\widehat{M} \int_V \nabla\textbf{u}\nabla\Delta\textbf{u}
    \end{split} \\
    & \leq \chi \Gamma \widehat{M} P_0 \left(\frac{\|\textbf{u}\|^2}{2}+\frac{\|\nabla \textbf{u}\|^2}{2}\right)+\Gamma \widehat{M} P_0 \left(\frac{\|\omega\|^2}{2}+\frac{\|\nabla \omega\|^2}{2}\right) + \chi^2\Ra^2c_P\epsilon\frac{\|\nabla\Delta\theta\|^2}{2}+ \frac{\|w\|^2}{2\epsilon}\\ & \quad + \Ra^2\epsilon\frac{\|\Delta\theta\|^2}{2} + \frac{\|\Delta w\|^2}{2\epsilon} +\chi^2\epsilon\Ra^2\frac{\|\Delta\theta\|^2}{2} + \frac{\|\Delta w\|^2}{2\epsilon} +\Ra^2c_P^3\epsilon\frac{\|\nabla\Delta\theta\|^2}{2}+ \frac{\|w\|^2}{2\epsilon}\\
    & \quad + \chi P_0\Gamma\widehat{M}\left(\frac{\|\textbf{u}\|^2}{2}+\frac{\|\nabla\Delta\textbf{u}\|^2}{2}\right)+\chi P_0\widehat{M}\left(\frac{\|\nabla\textbf{u}\|^2}{2}+\frac{\|\nabla\Delta\textbf{u}\|^2}{2}\right)\\
    & \leq \chi\Gamma \widehat{M} P_0 \left(\frac{1}{P_0\widehat{m}}+1\right)\frac{D_2}{2}+\Gamma \widehat{M} P_0 \left(\frac{1}{P_0\widehat{m}}+1\right)\frac{D_1}{2\lambda_2} +
    \chi^2\epsilon\Ra^2c_P \frac{D_2}{2}+\frac{1}{P_0\widehat{m}}\frac{D_1}{2c_P\lambda_1\epsilon } \\ & \quad+ \epsilon\Ra^2c_P \frac{D_2}{2}+\frac{D_1}{2\lambda_1\epsilon}+\chi^2\epsilon\Ra^2c_P \frac{D_2}{2}+\frac{D_1}{2\lambda_1\epsilon} +\epsilon\Ra^2  c_P^3\frac{D_2}{2}+ \frac{1}{P_0\widehat{m}}\frac{D_1}{2c_P\lambda_1\epsilon}\\
    &
    \quad +\chi P_0\Gamma\widehat{M}\left(c_P+1\right)\frac{D_2}{2}+2\chi P_0\widehat{M}\frac{D_2}{2}\\
    & = \chi\Gamma \widehat{M} P_0 \left(\frac{1}{P_0\widehat{m}}+2+c_P+\frac{2}{\Gamma}\right) \frac{D_2}{2} + \frac{D_1}{2} K + \epsilon\Ra^2c_P(2\chi^2+1+c_P^2) \frac{D_2}{2}
\end{aligned}
\end{equation}
where the constant $K$ is:
\begin{equation}
    K= \Gamma \widehat{M} P_0 \left(\frac{1}{P_0\widehat{m}}+1\right)\frac{1}{\lambda_2} + \frac{2}{\epsilon P_0 c_P\lambda_1\widehat{m}} + \frac{2}{\epsilon\lambda_1}
\end{equation}
Hence, upon choosing
\begin{equation}
    \chi=\left[2\Gamma \widehat{M} P_0 \left(\frac{1}{P_0\widehat{m}}+2+c_P+\frac{2}{\Gamma}\right)\right]^{-1}
\end{equation}
from the last inequality in \eqref{eq33}, it follows:
\begin{equation}\label{eq45}
    I_1 \leq \frac{D_2}{4}+\frac{D_1}{2}K+ \epsilon\Ra^2c_P(2\chi^2+1+c_P^2) \frac{D_2}{2} = \frac{D_2}{2} + \frac{D_1}{2}\bar{K}
\end{equation}
In \eqref{eq45}, the last equality holds upon choosing
\begin{equation}
\epsilon=\bar\epsilon:=\left(2\Ra^2c_P(2\chi^2+1+c_P^2)\right)^{-1}
\end{equation}
and, consequently,
\begin{equation}\label{eq49}
    \bar{K}=K(\bar\epsilon):= \Gamma \widehat{M} P_0 \left(\frac{1}{P_0\widehat{m}}+1\right)\frac{1}{\lambda_2} + \frac{2}{\bar\epsilon P_0c_P\lambda_1\widehat{m}} + \frac{2}{\bar\epsilon\lambda_1}
\end{equation}
By defining in \eqref{eq45}
\begin{equation}\label{eq50}
    b= \frac{1-m}{\bar{K}}
\end{equation}
and $D_3$ as in Eq. \eqref{d3}, then $bI_2\leq  D_3$.

\end{proof}

From the previous results, inequality \eqref{ineq1} turns into:
\begin{equation}\label{ineq2}
    \frac{dE}{dt} \leq -D_3+N_1+bN_2
\end{equation}

The nonlinear terms $N_1, N_2$ can be controlled as shown in Appendix \ref{appendix2}. Hence, the following inequality holds
\begin{equation}\label{ineq3}
    N_1+bN_2\leq A_0 D_3 \sqrt{E}
\end{equation}
where $A_0$ is a constant explicitly determined in Appendix \ref{appendix2}.

\subsection{Exponential decay of the energy}

From \eqref{eq49}-\eqref{eq50}, the evolution equation leads to the following inequality: 
\begin{equation}\label{energy_final}
    \frac{dE}{dt}\leq -D_3(1-A_0\sqrt{E})
\end{equation}
This estimate allows us to write the following conditional nonlinear stability theorem

\begin{thm}
    If $0<m<1$ and $E(0)<A_0^{-2}$, with $A_0$ given in Eq. \eqref{a0}, then from Eq. \eqref{energy_final}
    \begin{equation}\label{expdecay}
    E(t)\leq E(0) \,\mathrm{exp}\{-\Gamma_1(1-A_0\sqrt{E(0)})t\}
\end{equation}
with $\Gamma_1$ positive constant. Consequently, the motionless steady state $m_b$ is nonlinearly stable.
\end{thm}
\begin{proof}
    By applying the Poincar\'e inequality to $E_1$ and $E_2$, it turns out that
\begin{equation}
    \begin{aligned}
        & E_1\leq \xi D_1 \\
        & E_2\leq \xi D_2 
    \end{aligned}
\qquad \xi=\max\left\{ \frac{\Da}{2\widehat{m}}, \frac{c_P}{2},\frac{\Pr c_P}{2}\right\}
\end{equation}
Therefore, 
\begin{equation}\label{th2}
    E=E_1+bE_2 \leq \xi D_1 + b\xi D_2 \leq D_3\left(\frac{2\xi}{(1-m)} + 2\xi\right)
\end{equation}

Moreover, condition $E(0)<A_0^{-2}$ and Eq. \eqref{energy_final} ensure that 
\begin{equation}
    \frac{dE(0)}{dt}<0
\end{equation}
In a neighborhood of $t=0$, the energy is still decreasing in time. As a consequence, from a recursive argument, it follows that
\begin{equation}\label{th1}
    \frac{dE}{dt} \leq -D_3\left(1-A_0\sqrt{E(0)}\right) \qquad \forall t\geq 0
\end{equation}

Hence, from Eqs. \eqref{th2}-\eqref{th1}, we get
\begin{equation}
    \frac{dE}{dt}\leq -\Gamma_1 E\left(1-A_0\sqrt{E(0)}\right)
\end{equation}
where 
\begin{equation}
    \Gamma_1=\left(\frac{2\xi}{(1-m)} + 2\xi\right)^{-1}
\end{equation}
from which Eq. \eqref{expdecay} follows upon applying Gronwall's inequality.

\end{proof}

\subsection{The Euler-Lagrange equations}
In the present section, the variational problem Eq. \eqref{varprob} is presented and solved, and  the critical Rayleigh number for nonlinear stability of the conduction solution is determined.   
The Euler-Lagrange equations for Eq. \eqref{varprob} are
\begin{equation}\label{el}
    \begin{cases}
        -\lambda_1 \varphi +q\Ra^2 \Delta_1\varphi -\mathcal{T}^2 \varphi_{,zz} + \mu \Ra \Delta_1\Delta_1\theta -\lambda_2\mathcal{T}\omega_{,z} -2\lambda_1 \Delta\Delta w + 2f^\prime \lambda_1 P_0 w_{,z} + 2 f \lambda_1 P_0 \Delta w=0\\
        -\lambda_1 w + q \Ra^2 \Delta_1 w -\mathcal{T}^2 w_{,zz} -\Ra(1-q) \Delta\Delta_1\theta +\mathcal{T}P_0 (f \omega_{,z} + f^\prime \omega) +2\Delta \varphi=0\\
        -\Ra(1-q)\Delta\Delta_1 \varphi +\mu \Ra \Delta_1\Delta_1 w +2\mu \Delta\Delta_1\Delta_1 \theta=0\\
        -\mathcal{T}P_0 f\varphi_{,z} +\lambda_2 \mathcal{T}w_{,z} +2\lambda_2 \Delta\omega -2\lambda_2 P_0 f \omega=0
    \end{cases}
\end{equation}
Following a standard analysis with normal modes solutions Eq. \eqref{visc_expansion}, system \eqref{el} can be written as follows
\begin{equation}\label{el2}
    \begin{cases}
        -\lambda_1 \varphi -q\Ra^2 k^2\varphi -\mathcal{T}^2 D^2\varphi + \mu \Ra k^4\theta -\lambda_2\mathcal{T}D\omega -2\lambda_1 (D^2-k^2)^2 w + 2f^\prime \lambda_1 P_0 Dw + 2 f \lambda_1 P_0 (D^2-k^2) w=0\\
        -\lambda_1 w - q \Ra^2 k^2 w -\mathcal{T}^2 D^2w +\Ra(1-q) k^2(D^2-k^2)\theta +\mathcal{T}P_0 (f D\omega + f^\prime \omega) +2(D^2-k^2) \varphi=0\\
        \Ra(1-q)(D^2-k^2) \varphi +\mu \Ra k^2 w +2\mu k^2(D^2-k^2) \theta=0\\
        -\mathcal{T}P_0 fD\varphi +\lambda_2 \mathcal{T}Dw +2\lambda_2 (D^2-k^2)\omega -2\lambda_2 P_0 f \omega=0
    \end{cases}
\end{equation}

The critical Rayleigh number will be provided by solving the following problem
\begin{equation}\label{maxprob}
\Ra_{E}^2=\max_{\lambda_1,\lambda_2,\mu>0}\min_{k^2\in\mathbb{R}^+}\Ra^2(\lambda_1,\lambda_2,\mu,k^2)
\end{equation}
The difficulty of solving this problem lies in the fact that we cannot recover an analytical expression for $\Ra(\lambda_1,\lambda_2,\mu,k^2)$. 
Therefore, in order to solve the minmax problem, we employed a combination of the Chebyshev-$\tau$ method and the golden section method. 

The latter is classified as numerical optimization gradient-free scheme, because it does not require information about the derivatives of the objective function. 
The idea of this method consists of recursively eliminating a portion of the region where the minimum has been searched. At each iteration, the region is reduced by exactly $38.2\%$, which is higher than many other methods, according to \cite{extensiongolden}. This is consequence of the choice of the so-called golden number $\tau$ which comes from solving $\tau^2=1-\tau$ and is equal to $0.618$. Moreover, at each iteration, the method involves only $2^n$ function evaluations, where $n$ is the number of independent variables, in our case $n=3$ for the maximum problem and $n=1$ for the minimum one. For the problem Eq. \eqref{maxprob}, each function evaluation corresponds to the determination of the Rayleigh number through a single run of the Chebyshev-$\tau$ method. 

\section{Numerical results}\label{sec4}
The present section addresses the numerical resolution of differential eigenvalue problems arising from the linear and nonlinear stability analyses via the numerical methods proposed so far.
In particular, it is discussed the influence on the onset of instability of two competing effects such as variable viscosity and rotation. 

In \cite{CaponeGentile2000}, the authors investigated the onset of instability in a uniformly rotating porous medium and proved analytically the validity of the principle of exchange of stabilities when $\Pr\geq 1$. 
This result led us to analyse the effect of viscosity variations on the critical thresholds only in the diffusive regime for $\Pr\geq 1$, with the aim of understanding whether the principle is still satisfied or not and give an insight about the interplay between two competing effects, namely variable viscosity and rotation.

First of all, in a case where the fluid viscosity is constant, i.e. $\Gamma=0$, one would expect to recover the analytical threshold found in \cite{CaponeGentile2000}, i.e. 
\begin{equation}
    \Ra^2_L = \min_{k^2} \frac{(\pi^2 +k^2) \left[(\pi^2+k^2)(\pi^2+k^2+P_0)^2 + \pi^2\mathcal{T}^2\right]}{k^2(\pi^2+k^2+P_0)} 
\end{equation}

In Table \ref{tab2} it is displayed a comparison between the critical thresholds coming from the present linear and nonlinear stability analyses and the results obtained in \cite{CaponeGentile2000}. Coincidence of results is evident and proves the validity of numerical methods employed in the present paper.
\begin{table}[h!]\centering
 \scriptsize
		\begin{tabular}{@{}ll|llll|llll|lll@{}}\toprule
	\multicolumn{3}{c}{$\mathcal{T}=1$} & \phantom{c}&\multicolumn{3}{c}{$\mathcal{T}=10$} &  \phantom{c}&\multicolumn{3}{c}{$\mathcal{T}=5$}&\phantom{c}& \\
			\cmidrule{1-3} \cmidrule{5-7} \cmidrule{9-11}
	$\Ra_L \cite{CaponeGentile2000} $ &	$\Ra_{E,c}$&$ \Ra_{_L} $&& $\Ra_L \cite{CaponeGentile2000}$ &	$\Ra_{E,c}$&$ \Ra_{_L} $  && $\Ra_L \cite{CaponeGentile2000}$ &	$\Ra_{E,c}$&$ \Ra_{_L} $ &&$ P_0 $\\ \midrule
		 26.1066&  26.1061 & 26.1057 & &29.0253& 29.0033 & 29.0250 && 26.9267&  26.9340& 26.9257 && 0.5 \\
          26.5236&  26.5168 & 26.5236 & &29.3068& 29.2969 & 29.3061 && 27.2978& 27.3153 & 27.2974 && 1 \\
          29.5984&  29.5775 & 29.5974 & &31.5758& 31.5751 & 31.5750 &&30.1194& 30.1182 & 30.1187 && 5 \\
			\bottomrule
		\end{tabular}
		\caption{Comparison with results known in literature from \cite{CaponeGentile2000} for quoted values of $\mathcal{T}$ and $P_0$, with $\Gamma=0$ and $\Pr=2$.}
		\label{tab2}
	\end{table}

\begin{figure}
    \centering
    \includegraphics[scale=0.7]{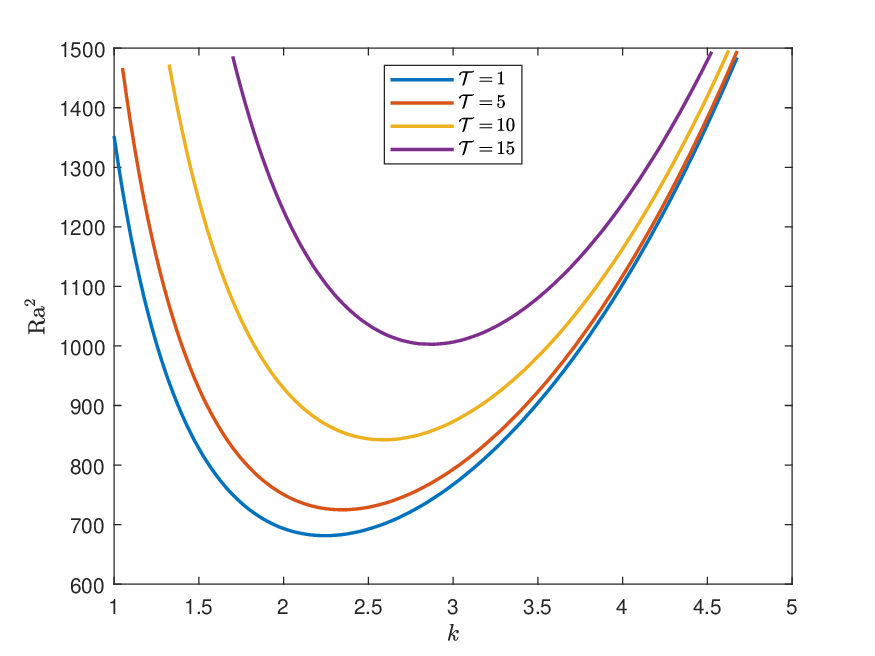}
    \caption{$\Da=2, \Pr=2, \Gamma=0$}
    \label{fig1}
\end{figure}

\begin{figure}
    \centering
    \includegraphics[scale=0.3]{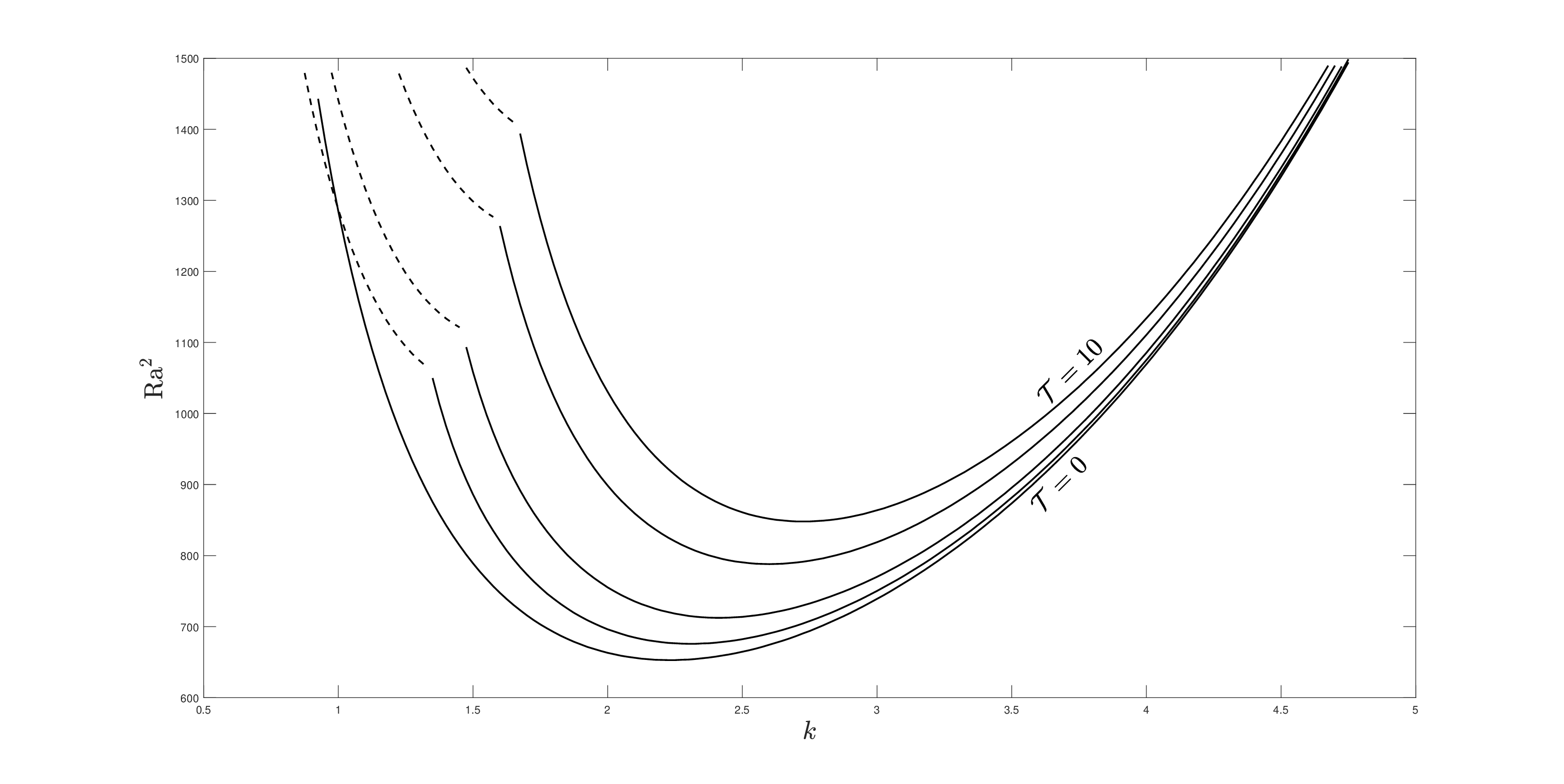}
    \caption{$\Da=2, \Pr=2, \Gamma=1.5$}
    \label{fig2}
\end{figure}

In order to understand the difference of the results between the two papers, in terms of the reality of leading eigenvalues of the respective differential operators, 
Figures \ref{fig1}-\ref{fig2} have been reported. Both describe the behaviour of neutral curves with respect to variations of the Taylor number $\mathcal{T}$. The former is useful to visualise the results from \cite{CaponeGentile2000}, computed from \eqref{modad2}-\eqref{bc5} imposing $\Gamma=0$, while the latter displays results from the present paper for $\Gamma=\frac{3}{2}$.

According to Figure \ref{fig1}, all the eigenvalues are real and only stationary convection can occur. Therefore, oscillating motions are not present when $\Gamma=0$, confirming the results about the validity of the principle of the exchange of stabilities obtained in \cite{CaponeGentile2000}. 

On the other hand, in Figure \ref{fig2}, the existence of modes that may lead to occurrence of travelling waves within the fluid is displayed. This is direct consequence of the fact that pure imaginary eigenvalues occur at the onset of instability.
Comparison between the two pictures shows that the assumption of non-constant viscosity implies that the principle of exchange of stabilities does not hold, as long as $\Gamma$ is large enough. This result is remarkable if compared to what is available in literature (see for example \cite{torrance, straughan1986, kassoy1975, acta2022, joseph, rajagopal2009stability}), where to the best of our knowledge, the assumption of variable viscosity does not impact on the existence of new oscillatory modes. 
Another aspect emerging from Figure \ref{fig2} regards the stabilising effect of rotation on the onset of instability. It is evident that the critical Rayleigh number increases for increasing $\mathcal{T}$. This mechanism is well-known in literature and it has been already presented in the introduction. 
Moreover, also the critical point where neutral curves for steady and oscillatory convection meet is moving right, meaning that the region of modes that would produce oscillating motions is increasing, even though to obtain oscillatory convection a higher Rayleigh number is required.
The reader may also notice that all the neutral curves collapse for modes whose wavelength is small, meaning that for those modes the effect of increasing rotation is less remarkable.

\begin{figure}
    \centering
    \includegraphics[scale=0.7]{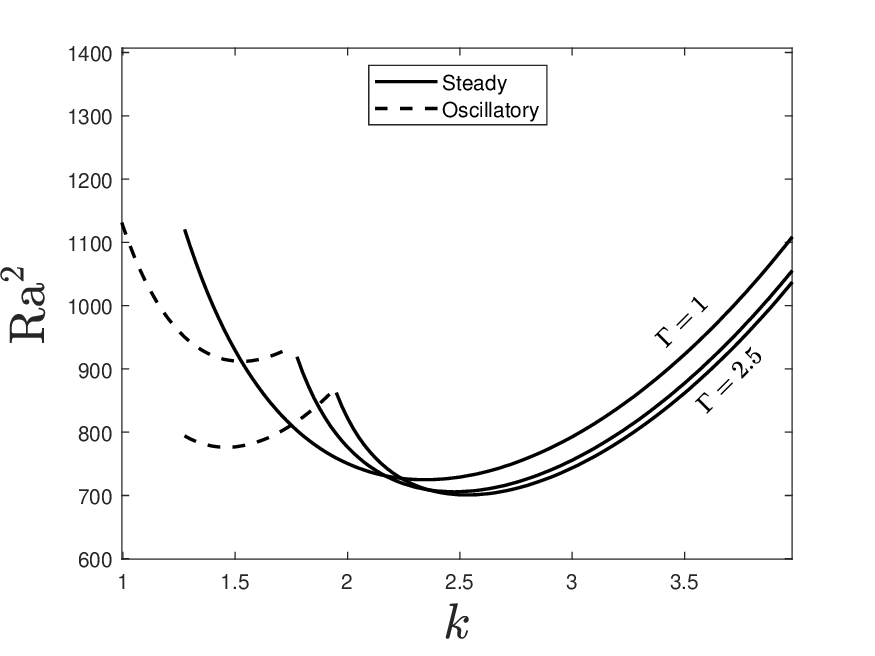}
    \caption{$\Da=2, \Pr=2, \mathcal{T}=5$}
    \label{fig3}
\end{figure}

\begin{figure}
    \centering
    \includegraphics[scale=0.3]{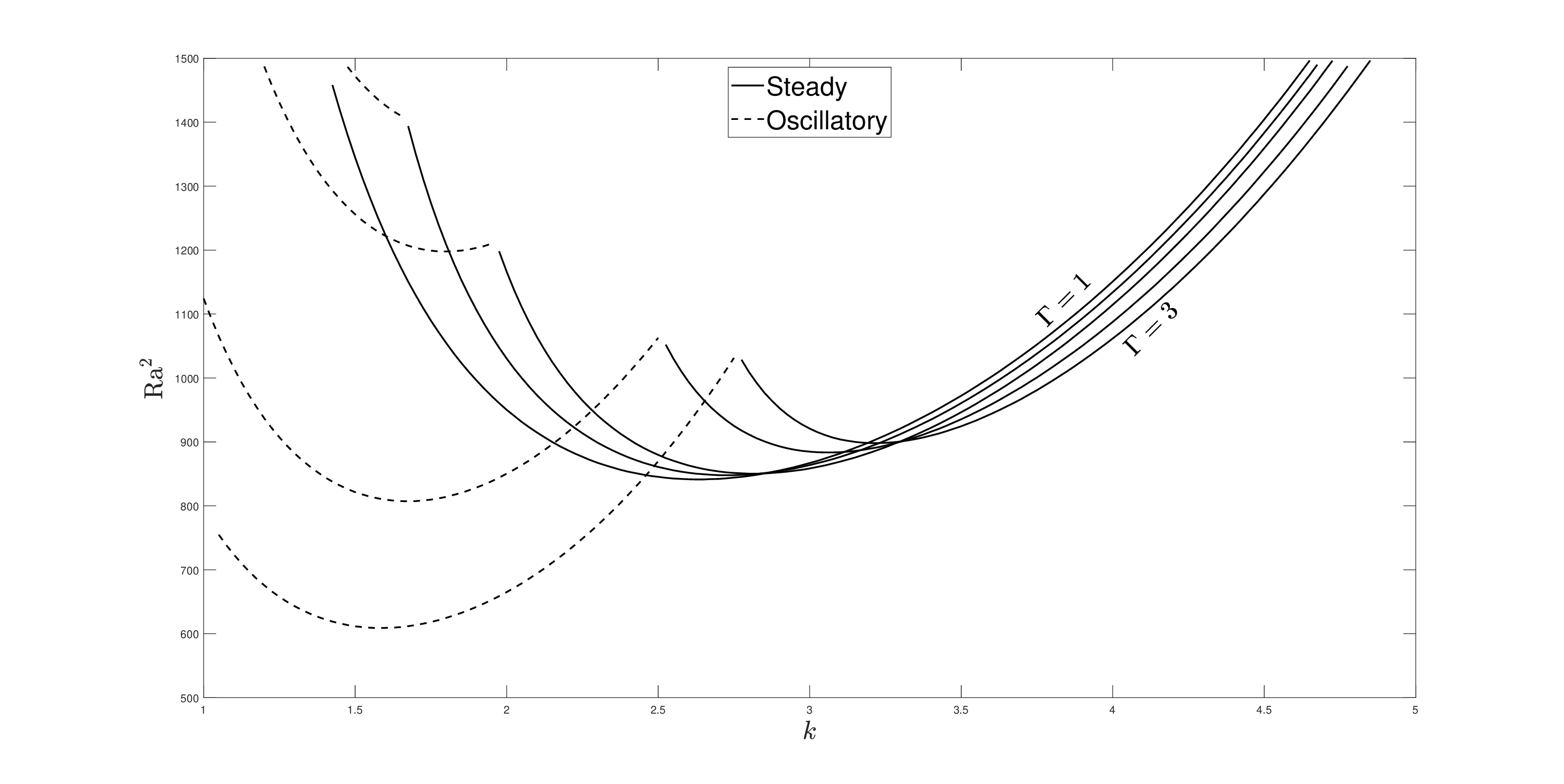}
    \caption{$\Da=2, \Pr=2, \mathcal{T}=10$}
    \label{fig4}
\end{figure}

Figures \ref{fig3}-\ref{fig4} show the behaviour of neutral curves with respect to increasing $\Gamma$, i.e. the gradient of viscosity, for two different speeds of rotation, i.e. $\mathcal{T}=5,10$. 
This choice is based on the idea of checking the behaviour of neutral curves in a regime at low and high speed of rotation. While the choice of physically reasonable values of $\Gamma$ is based on experimental data. For instance, for values of $\gamma$ reported in Table \ref{tab1}, if the temperature difference between top and bottom plane is $40^\circ C$, then $\Gamma=40\times0.065=2.6$.

In Figure \ref{fig3}, it is evident that both oscillatory and steady convection may be triggered by certain modes, even though absence of oscillating motions for $\Gamma=1$ is reported.
The decreasing behaviour of the Rayleigh number for increasing $\Gamma$ is preserved and known in literature. In this sense, the rate of rotation chosen does not affect the impact of variable viscosity.

Opposite trend is displayed in Figure \ref{fig4}, where $\mathcal{T}=10$. Indeed, the critical Rayleigh number is increasing for increasing $\Gamma$. This trend can be physically explained because greater $\Gamma$ implies that the fluid struggles less to move at the bottom of the layer, in the proximity of the hot plane. A speed of rotation large enough takes advantage of the fluid freedom pushing particles outward in the horizontal direction. Therefore, the interplay between rotation and viscosity results in a delay of the onset of convection.
Interestingly, the usual trend is preserved for small-wavelength perturbations, whose critical Rayleigh number is decreasing for increasing $\Gamma$. Moreover, as in Figure \ref{fig2}, oscillating motions are not present for small enough values of $\Gamma$, and the number of modes that triggers oscillating motions is larger for increasing $\Gamma$. 

In Table \ref{tab2} comparison between critical Rayleigh numbers from linear and nonlinear stability analysis is reported. For sufficiently low values of Taylor number and gradient of viscosity, difference between thresholds falls within tolerance and numerical error due to approximations in the numerical schemes employed. Whereas, increasing either $\mathcal{T}$ or $\Gamma$ results in a wider gap between stability thresholds, but still of the order of $3\%$.

\begin{table}[h!]\centering
 \scriptsize
		\begin{tabular}{@{}l|lll|lll|lll@{}}\toprule
	\multicolumn{2}{c}{$\mathcal{T}=1$} & \phantom{c}&\multicolumn{2}{c}{$\mathcal{T}=5$} &  \phantom{c}&\multicolumn{2}{c}{$\mathcal{T}=10$}&\phantom{c}& \\
			\cmidrule{1-2} \cmidrule{4-5} \cmidrule{7-8}
	$\Ra_{E,c}$&$ \Ra_{_L} $&& 	$\Ra_{E,c}$&$ \Ra_{_L} $  && 	$\Ra_{E,c}$&$ \Ra_{_L} $ &&$ \Gamma $\\ \midrule
		 26.0367 & 26.0486 & & 26.8661 & 26.8924 &&   28.9858 & 29.0261 && 0.5 \\
          25.8492&  25.9049 && 26.6997  & 26.8024 && 28.8328 &  29.0086 & & 1 \\
          25.0790&  25.2824 && 25.9583  & 26.5698 && 28.1340 & 29.1636 && 2 \\
			\bottomrule
		\end{tabular}
		\caption{Comparison of results from linear and nonlinear analysis for quoted values of $\mathcal{T}$ and $\Gamma$, with $P_0=0.5$ and $\Pr=2$.}
		\label{tab2}
	\end{table}

\section{Conclusions}
Our investigation aimed at providing new results on the interplay between two counter-acting effects in a fluid-saturated porous medium heated from below. In this paper, a variable viscosity fluid is considered within a uniformly rotating porous medium. Uniform rotation about the vertical axis is well-known to inhibit the onset of thermal convection while temperature-dependent fluid viscosity has opposite consequences.
From linear stability analysis, it is shown that while the effect of increasing rotation rate produces higher critical Rayleigh numbers, increasing the gradient of viscosity may lead to different trends of the critical linear threshold depending on the rate of rotation considered. More importantly, it has been proved numerically that considering a fluid with non-zero viscosity gradient may lead to the occurrence of convective patterns oscillating in time. This result is remarkable in the regime $\Pr\geq 1$, where, as proved in \cite{CaponeGentile2000}, in absence of variations of viscosity with temperature the onset of convection is always steady. Moreover, a consistent nonlinear stability analysis has been developed via the energy method. The choice of a suitable Lyapunov functional is very tricky in the context of rotating porous media, but in this case it has been possible to determine an energy functional that could lead to proximity between linear and nonlinear stability thresholds and exponential decay in time of the perturbations at low Rayleigh numbers, provided that the energy satisfies a constraint on the initial datum.

\appendix
\section{Appendix}
 \subsection{Numerical procedure}\label{appendix1}

The differential eigenvalue problem Eqs. \eqref{modad2}-\eqref{bc5} has been solved via the Chebyshev-$\tau$ method. The present section is devoted to provide some hints about the implementation of the method.

First step involves the introduction of an approximation of the unknown fields that are then written as truncated series of Chebyshev polynomials. Therefore, let $T_k\in L^2(-1,1)$ be the $k$-th Chebyshev polynomial, with $k\in\N_0$. 
Given that these polynomials are defined in a different interval compared to the physical domain of the problem at stake, transformation of the domain from $(0,1)$ to $(-1,1)$ is required. Now, solution of Eq. \eqref{modad2} can be expanded as truncated series:
\begin{equation}\label{expansion2}
 w= \sum_{k=0}^{N+2} W_k T_k(z), \quad \Phi = \sum_{k=0}^{N+2} \Phi_k T_k(z), \quad \Theta = \sum_{k=0}^{N+2} \Theta_k T_k(z), \quad \omega = \sum_{k=0}^{N+2} \omega_k T_k(z),
\end{equation}

The truncation introduces an error. Therefore, according to \cite{dongarra}, rather than solving the  differential problem 
\begin{equation}
    L(w,\Phi,\theta,\omega)=0
\end{equation}
where $L$ is the linear differential operator in Eq. \eqref{modad2}, the following system has to be solved:
\begin{equation}
    L(w,\Phi,\theta,\omega)=\sum_{i=1}^{n_b}\tau_i T_{N+i}
\end{equation}
where $\tau_i$ is a measure of the truncation error and $n_b$ is the total number of boundary conditions \cite{lanczos1988}.

It is well-known that Chebyshev polynomials are orthogonal with respect to the scalar product 
\begin{equation}
 \langle f, g\rangle := \int_{-1}^{1} \frac{f g}{\sqrt{1-z^2}} d z \qquad f, g \in L^2(-1,1)
\end{equation}
Therefore, by scalar multiplication of the resulting equations obtained upon substituting Eq. \eqref{expansion2} into \eqref{modad2},
 with $T_k$ for $k=0,\dots, N$, system \eqref{modad2}-\eqref{bc5} reduces to an algebraic eigenvalue problem that can be solved with a simple Matlab routine: 
\begin{equation}\label{modad3}
\begin{cases}
    0 =(4D^2-k^2) w -\Phi\\
    \sigma F\Phi = F(4D^2-k^2)\Phi -P_0\Phi -P_0 \Gamma 2Dw-\Ra F k^2 \theta - \mathcal{T} F 2D\omega\\
    \Pr\sigma \theta =\Ra w + (4D^2-k^2)\theta\\
    \sigma F \omega = F(4D^2-k^2) \omega -P_0 \omega +\mathcal{T} F 2D w    
\end{cases} 
\end{equation}
In Eq. \eqref{modad3}, $F$ is a matrix and represents the approximation of $f^{-1}(z)=\text{exp}(-\Gamma(z-1/2))$ via Chebyshev polynomials. The structure of this matrix is therefore recursive and reduces to the identity matrix when $\Gamma=0$. Details about the determination of such a matrix are reported in \cite{arnone2023chebyshev, acta2022}.

As a result, a generalised eigenvalue problem of this kind is obtained
\begin{equation}\label{eigprob}
    \mathcal{A} \textbf{X} = \sigma \mathcal{B} \textbf{X}
\end{equation}
where $\mathcal{A}, \mathcal{B} \in \mathcal{M}^{[4(N+1)+8]\times [4(N+1)+8]}$ are completed by adding eight equations that come from the boundary conditions. Hence, 
\begin{equation}
    \mathcal{A}= \begin{pmatrix}
        4D^2-k^2I & -I & \bm{0}&\bm{0}\\
        \text{b.c.} & 0\dots0 &  0\dots0 & 0\dots0\\
        -2\Gamma P_0 D & F(4D^2-k^2)-P_0I & -\Ra Fk^2 & -2\mathcal{T}FD\\
        0\dots0&\text{b.c.}&0\dots0&0\dots0\\
        \Ra I & \bm{0} & 4D^2-k^2I & \bm{0}\\
        0\dots0&0\dots0&\text{b.c.}&0\dots0\\
        2\mathcal{T}FD&\bm{0} & \bm{0} & F(4D^2-k^2I)-P_0I\\
        0\dots0&0\dots0&0\dots0&\text{b.c.}
    \end{pmatrix}
\end{equation}
where $\bm{0}\in\mathcal{M}^{N+1,N+1}$ and
\begin{equation}
    \mathcal{B}= \begin{pmatrix}
        \bm{0} &\bm{0} & \bm{0} & \bm{0}\\
        0\dots0&0\dots0&0\dots0&0\dots0\\
        \bm{0} & F & \bm{0}& \bm{0}\\
        0\dots0&0\dots0&0\dots0&0\dots0\\
        \bm{0} & \bm{0} & \Pr I& \bm{0}\\
        0\dots0&0\dots0&0\dots0&0\dots0\\
        \bm{0} & \bm{0} &  \bm{0} & F\\ 
       0\dots0&0\dots0&0\dots0&0\dots0\\
    \end{pmatrix}
\end{equation}

As the reader can notice, $\mathcal{B}$ is singular, mainly because of the boundary conditions equations. Consequently, spurious eigenvalues can be detected. This problem has been frequently addressed in literature, see \cite{bourne2003,  gheorghiu2014spectral, mcfadden}. 
Some details about how to fix this issue are reported in \cite{arnone2023onset} and references therein.

The leading eigenvalue related to instability can be safely determined, once some algebraic manipulations in matrices $\mathcal{A}, \mathcal{B}$ are completed.

\subsection{Control of nonlinearities}\label{appendix2}
In this section, functional analysis tools are employed in order to control nonlinear terms $N_1$, $N_2$ in \eqref{ineq2}.
These terms are controlled with the product between $E$ and $D_3$, to obtain inequality \eqref{ineq3}.

Let us introduce some preliminary inequalities that will help us in estimating the nonlinear terms in the evolution equation. From the Sobolev embedding inequalities, there exists a positive constant $C$ depending on the periodicity cell $V$ such that (see \cite{adams})
\begin{equation}
    \sup_V|\textbf{u}|\leq C \|\Delta\textbf{u}\| \qquad \sup_V|\nabla\theta|\leq C\|\nabla\Delta\theta\| \qquad \sup_V|\theta|\leq C\|\Delta\theta\|
\end{equation}
We refer the reader to \cite{mulone1997} for a more detailed discussion about the validity of these inequalities and the definition of the constant $C$.

Moreover, in the following many well-known inequalities such as Poincar\'e inequality or Cauchy-Schwartz inequality are employed. In \cite{RioneroMulone1987}, we found the following helpful estimate is provided
\begin{equation}
    \|w_{x^ix^j}\|\leq \|\Delta w\|
\end{equation}
which originally has been proved in \cite{adams}.

Given that
\begin{equation}
\begin{aligned}
    & \nabla\times\nabla\times(f(z) \theta \textbf{u})\cdot\textbf{k} = \partial_{,z}(\nabla \cdot f\theta\textbf{u}) - \Delta (f\theta\textbf{u})=\\
    & f^{\prime}\textbf{u}\nabla_1\theta + f\textbf{u}_{,z}\nabla_1 \theta - f w_{,z}\theta_{,z} + f\textbf{u}\cdot \nabla\theta_{,z} -f^\prime\theta w_{,z} -f w \Delta_1 \theta -f \theta \Delta w -2f \nabla_1 \theta \cdot \nabla_1 w 
    \end{aligned}
\end{equation}
and
\begin{equation}
    \nabla\times( f(z) \theta \textbf{u})\cdot\textbf{k} = f \nabla_1 \theta \cdot (\textbf{u}\times \textbf{k}) + f \theta \omega
\end{equation}
$N_1$ can be rearranged and each term can be estimated as follows, upon defining 
\begin{equation}
    J=\widehat{M}C\left(\frac{2}{b}\right)^{\frac{3}{2}}
\end{equation}
recalling that $\widehat{M}:=\max_{[0,1]}f(z)$.

\newpage
Therefore,
\begin{equation}
    \begin{aligned}
    & (f^{\prime}\textbf{u}\cdot\nabla_1\theta, w) \leq \Gamma \widehat{M} \sup|\nabla\theta| \|w\| \|\textbf{u}\| \leq \Gamma c_P J D_2 \sqrt{E_2}\\
    & ( f\textbf{u}_{,z}\nabla_1 \theta ,w) \leq \widehat{M} \sup|\nabla \theta| \|\nabla\textbf{u}\| \|w\| \leq J D_2 \sqrt{E_2}\\ 
    & (- f w_{,z}\theta_{,z} ,w) \leq \widehat{M} \sup|\nabla \theta| \|\nabla w\| \|w\|\leq JD_2 \sqrt{E_2}\\ 
    & ( f\textbf{u}\cdot \nabla\theta_{,z},w) = -(f^\prime \textbf{u}\cdot \nabla\theta w+ f\textbf{u}_{,z}\cdot \nabla\theta w + f \textbf{u} \nabla\theta w_{,z} ) \\
    &\qquad\qquad\qquad\leq \sup|\nabla\theta| (\Gamma \widehat{M} \|\textbf{u}\|\|w\| + \widehat{M} \|\nabla\textbf{u}\|\|w\| + \widehat{M} \|\textbf{u}\|\|\nabla w\|)\leq J(2+\Gamma c_P)D_2 \sqrt{E_2}\\
    & ( -f^\prime\theta w_{,z},w) \leq \Gamma \widehat{M} \sup|\theta| \|\nabla w\| \|w\|\leq \Gamma c_p J\\
    & ( -f w \Delta_1 \theta,w) \leq \widehat{M} \sup|\textbf{u}| \|w\| \|\Delta_1\theta\|\leq J c_P \sqrt{\frac{1}{\chi}}\\
    & ( -f \theta \Delta w,w) \leq  \widehat{M} \sup|\theta| \|w\| \|\Delta w\| \leq J c_P \sqrt{\frac{1}{\chi}}\\
    & ( -2f \nabla_1 \theta \cdot \nabla_1 w ,w) \leq 2 \widehat{M} \sup|\nabla\theta| \|w\| \|\nabla\textbf{u}\|\leq 2J\\    
            & \begin{split}
                {\color{black}-\mu \Pr (\Delta_1(\textbf{u}\cdot \nabla\theta), \Delta_1\theta)} 
            & = \mu \Pr(\Delta_1\textbf{u}\cdot\nabla\theta + \textbf{u}\cdot \nabla\Delta_1\theta+2\nabla_1\textbf{u}\nabla\nabla_1\theta,\Delta\theta) \\
            & \leq\mu \Pr(\sup|\nabla\theta|\|\Delta\textbf{u}\|\|\Delta\theta\|+\sup|\textbf{u}|\|\nabla\Delta_1\theta\|\|\Delta\theta\|+2\sup|\nabla_1\textbf{u}|\|\Delta\theta\|\|\Delta\theta\|) \\
            & \leq \left(\frac{C}{\sqrt{\chi\Pr}} +\frac{C}{\sqrt{\chi\Pr}}+2c_P\frac{ C}{\sqrt{\chi\Pr}}\right)\mu\Pr \left(\frac{2}{b}\right)^{\frac{3}{2}}D_2\sqrt{E_2}
            \end{split}\\
        &
        \begin{split}
            {\color{black}-(\Ra\Delta_1(\textbf{u}\cdot\nabla\theta),\varphi)} & \leq\Ra (\Delta_1\textbf{u}\cdot\nabla\theta + \textbf{u}\cdot \nabla\Delta_1\theta+2\nabla_1\textbf{u}\nabla\nabla_1\theta,\varphi)\\
            & \leq \Ra(\sup|\nabla\theta|\|\Delta\textbf{u}\|\|\varphi\|+\sup|\textbf{u}|\|\nabla\Delta_1\theta\|\|\varphi\|+2\sup|\nabla_1\textbf{u}|\|\Delta\theta\|\|\varphi\| )\\
            & \leq\Ra \left(\frac{C}{\sqrt{\chi}}\frac{2\sqrt{2}}{b} +\frac{C}{\sqrt{\chi}}\frac{2\sqrt{2}}{b} + 2c_P\frac{C}{\sqrt{\chi}}\frac{2\sqrt{2}}{b} \right) D_2 \sqrt{E_1}
        \end{split}\\
        & 
        \begin{split}
            ( \partial_{,z}\nabla\times(f(z) \theta \textbf{u})\cdot\textbf{k},\varphi) & = (\partial_{,z}(f \nabla_1 \theta \cdot (\textbf{u}\times \textbf{k})) + \partial_{,z}(f \theta \omega), \varphi)\\ & \leq \widehat{M}\|\varphi\|(\Gamma  \sup|\nabla_1 \theta| 
            \|\textbf{u}\| + \sup|\textbf{u}\times \textbf{k}| \|\Delta\theta\| +  \sup|\nabla_1\theta|\|\textbf{u}_{,z}\|\\
            &\quad +\Gamma  \sup|\theta|\|\omega\|+ \sup|\theta|\|\nabla\omega\|+\sup|\nabla\theta|\|\omega\| )\\
            & \leq \widehat{M} C\frac{2}{b}\left(\frac{\Gamma}{\sqrt{P_0\widehat{m}}}+ \frac{c_P}{\sqrt{\chi}}+1+\frac{\Gamma c_P}{\sqrt{P_0\widehat{m}}}+ c_P + \frac{1}{\sqrt{P_0\widehat{m}}}\right)D_2 \sqrt{E_1}
        \end{split} \\
    & (\nabla\times( f(z) \theta \textbf{u})\cdot\textbf{k},\omega) \leq \widehat{M} \sup|\nabla \theta| \|\omega\| \|\textbf{u}\| + \widehat{M} \sup|\theta| \|\omega\| \|\omega\|\leq\left( c_P J + c_P \frac{1}{\sqrt{P_0\widehat{m}}}J\right)D_2\sqrt{E_2}\\
    \end{aligned}
\end{equation}
and $N_2$
\begin{equation}
    \begin{aligned}
        & \hat{\gamma}(f\theta\textbf{u},\Delta\textbf{u}) \leq \hat{\gamma}
        \widehat{M} \sup |\theta| \|\textbf{u}\| \|\Delta\textbf{u}\|\leq \hat{\gamma}J c_P\sqrt{\frac{1}{\chi}}D_2 \sqrt{E_2}\\
        & 
        \begin{split}
            {\color{black}-\Pr(\Delta(\textbf{u}\cdot\nabla\theta),\Delta\theta)} &\leq \Pr(\Delta\textbf{u}\cdot\nabla\theta+\textbf{u}\cdot\Delta\nabla\theta+2\nabla\textbf{u}\nabla\nabla\theta, \Delta\theta)
        \\
        & \leq \Pr(\sup|\nabla\theta|\|\Delta\textbf{u}\| \|\Delta\theta\|+\sup|\textbf{u}|\|\Delta\nabla\theta\|\|\Delta\theta\|+2\sup|\nabla\textbf{u}| \|\Delta\theta\|\|\Delta\theta\|)\\
        & \leq \Pr \left(\frac{c_P}{\sqrt{\chi}}+\frac{c_P}{\sqrt{\chi}}+\frac{2}{\sqrt{\Pr}}\frac{c_P}{\sqrt{\chi}}\right)\frac{J}{M}D_2 \sqrt{E_2}
        \end{split}\\
        & 
        \begin{split}
            (\nabla\times( f(z) \theta \textbf{u})\cdot\textbf{k},\Delta\omega) & = (f \nabla_1 \theta \cdot (\textbf{u}\times \textbf{k}) + f \theta \omega, \Delta \omega)\\ &\leq \widehat{M} \sup|\nabla \theta| \|\textbf{u}\| \|\Delta\omega\| + \widehat{M}\sup|\theta| \|\omega\| \|\Delta \omega\|\leq J (1+c_P)D_2\sqrt{E_2}
        \end{split}\\
        &(f\theta\textbf{u},\textbf{u})\leq \widehat{M}\sup|\theta| \|\textbf{u}\| \|\textbf{u}\|\leq c_P^2 JD_2\sqrt{E_2}\\
        & 
        \begin{split}
            (\nabla\times( f(z) \theta \textbf{u})\cdot\textbf{k},\omega) & =(f \nabla_1 \theta \cdot (\textbf{u}\times \textbf{k}) + f \theta \omega,  \omega)\\ & \leq \widehat{M} \sup|\nabla \theta| \|\textbf{u}\| \|\omega\| + \widehat{M} \sup|\theta| \|\omega\| \| \omega\|\leq  J\left(c_P+\frac{1}{\sqrt{P_0\widehat{m}}}\right)D_2\sqrt{E_2}
        \end{split}\\
     &- \chi (f^\prime \theta \textbf{u}, \nabla\Delta\textbf{u})\leq \chi \Gamma \widehat{M} \sup|\theta| \|\textbf{u}\|\|\nabla\Delta\textbf{u}\|\leq \chi c_P \Gamma J D_2 \sqrt{E_2}\\
     & -\chi (f\nabla\theta\textbf{u},\nabla\Delta\textbf{u})\leq \chi \widehat{M} \sup|\nabla\theta| \|\textbf{u}\|\|\nabla\Delta\textbf{u}\| \leq \chi  J D_2 \sqrt{E_2}\\
     & -\chi (f\theta \nabla\textbf{u}, \nabla\Delta\textbf{u})\leq \chi  \widehat{M} \sup|\theta| \|\nabla\textbf{u}\|\|\nabla\Delta\textbf{u}\|\leq\sqrt{\chi}  J D_2 \sqrt{E_2}
    \end{aligned}
\end{equation}

Finally, $A_0$ assumes the following expression:
\begin{equation}\label{a0}
\begin{split}
    A_0 & =\lambda_1 \hat{\gamma}\left(3\Gamma c_P +6 +\frac{2c_P}{\sqrt{\chi}}\right)J+ \frac{2(1+c_P)}{\chi \widehat{M}}\mu\sqrt{\Pr}J+ \frac{2(1+c_P)}{\chi \widehat{M}}\Ra\sqrt{b} J \\
    &+ \mathcal{T}\hat{\gamma}\left(\frac{\Gamma(c_P+1)+1}{\sqrt{P_0\widehat{m}}}+ \frac{c_P}{\sqrt{\chi}}+1+ c_P \right)\sqrt{\frac{2}{b}}J + \lambda_2\hat\gamma\left(\frac{1}{\sqrt{P_0\widehat{m}}}+1\right)c_PJ\\
    & +\sqrt{\chi} P_0 \hat{\gamma}c_P J+ \Pr C c_P \frac{2}{\widehat{M}\sqrt{\chi}}\left(1+\frac{1}{\sqrt{\Pr}}\right)J + P_0 \hat\gamma \left((1+c_P)^2 + \frac{1}{\sqrt{P_0\widehat{m}}}\right)J\\
    & +\chi\left(\Gamma c_P+1 + \frac{1}{\sqrt{\chi}}\right)J 
\end{split}
\end{equation}

\bibliographystyle{unsrt}
\bibliography{mybib}

\end{document}